\documentclass[11pt]{article}
\usepackage{amsthm}
\usepackage{amssymb}
\usepackage{amsmath}
\usepackage{mathtools}
\usepackage{graphicx}
\usepackage{colortbl}
\usepackage{cite}
\usepackage{authblk}
\usepackage{algorithm}
\usepackage{algorithmicx}
\usepackage[noend]{algpseudocode}
\usepackage{enumitem}

\usepackage[textwidth=2cm]{todonotes}

\colorlet{green}{green!50!black}
\usepackage[breaklinks=true]{hyperref} 
\hypersetup{
	colorlinks=true,
	citecolor=green,
	linkcolor=blue,  
	urlcolor=black,
}

\begin{document}
	
	\def\Nset{\mathbb{N}}
	\def\Ascr{\mathcal{A}}
	\def\Bscr{\mathcal{B}}
	\def\Cscr{\mathcal{C}}
	\def\Dscr{\mathcal{D}}
	\def\Escr{\mathcal{E}}
	\def\Fscr{\mathcal{F}}
	\def\Hscr{\mathcal{H}}
	\def\Iscr{\mathcal{I}}
	\def\Lscr{\mathcal{L}}
	\def\Mscr{\mathcal{M}}
	\def\Nscr{\mathcal{N}}
	\def\Pscr{\mathcal{P}}
	\def\Qscr{\mathcal{Q}}
	\def\Rscr{\mathcal{R}}
	\def\Sscr{\mathcal{S}}
	\def\Tscr{\mathcal{T}}
	\def\Uscr{\mathcal{U}}
	\def\Wscr{\mathcal{W}}
	\def\Xscr{\mathcal{X}}
	\def\cupp{\stackrel{.}{\cup}}
	\def\bold{\bf\boldmath}
	
	\newcommand{\rouge}[1]{\textcolor{red}{\tt \footnotesize #1}}
	\newcommand{\boldheader}[1]{\smallskip\noindent{\bold #1:}\quad}
	\newcommand{\PP}{\mbox{\slshape P}}
	\newcommand{\NP}{\mbox{\slshape NP}}
	\newcommand{\RP}{\mbox{\slshape RP}}
	\newcommand{\coNP}{\mbox{\textsc{co}\slshape NP}}
	\newcommand{\opt}{\mbox{\scriptsize\rm OPT}}
	\newcommand{\ec}{\mbox{\scriptsize\rm OPT}_{\small\rm 2EC}}
	\newcommand{\lp}{\mbox{\scriptsize\rm LP}}
	\newcommand{\inn}{\mbox{\rm in}}
	\newcommand{\deff}{\mbox{\rm sur}}
	\newcommand{\MAXSNP}{\mbox{\slshape MAXSNP}}
	\newtheorem*{theorem*}{Theorem}
	\newtheorem{theorem}{Theorem}
	\numberwithin{theorem}{section}
	\newtheorem*{lemma*}{Lemma}
	\newtheorem{lemma}[theorem]{Lemma}
	\newtheorem{conjecture}[theorem]{Conjecture}
	\newtheorem{problem}[theorem]{Problem}
	\newtheorem*{corollary*}{Corollary}
	\newtheorem{corollary}[theorem]{Corollary}
	\newtheorem{proposition}[theorem]{Proposition}
	\newtheorem{definition}[theorem]{Definition}
	
	%To make the "Proof." header bold:
	\makeatletter
	\renewenvironment{proof}[1][\proofname]{\par
		\pushQED{\qed}%
		\normalfont \topsep6\p@\@plus6\p@\relax
		\trivlist
		\item\relax
		{\bfseries  %<- font shape
			#1\@addpunct{.}}\hspace\labelsep\ignorespaces %<- includes punctuation code
	}{%
		\popQED\endtrivlist\@endpefalse
	}
	\makeatother

	% For placing the colon correctly (i.e., vertically centered) in \leteq 
	\def\leteq{\vcentcolon =}
	\def\cupp{\stackrel{.}{\cup}}
	%\def\NP{\mathsf{NP}}
	%\def\coNP{\textsc{coNP}}
	%\def\RP{\textsc{RP}}
	
	%\definecolor{orange}{rgb}{1,0.9,0}
	%\definecolor{violet}{rgb}{0.8,0,1}
	%\definecolor{darkgreen}{rgb}{0,0.5,0}
	%\definecolor{grey}{rgb}{0.75,0.75,0.75}
	
	\title{Odd Paths, Cycles and $T$-joins:\\ Connections and Algorithms}
	\author[1,2]{Ildik\'o Schlotter} 
	\author[3]{Andr\'as Seb\H o}
	\affil[1]{Centre for Economic and Regional Studies, Hungary; \texttt{schlotter.ildiko@krtk.hu}}
	\affil[2]{Budapest University of Technology and Economics, Hungary}
	\affil[3]{CNRS, Univ. Grenoble Alpes, France; \texttt{andras.sebo@grenoble-inp.fr}}
	\date{}

	\maketitle
	\begin{abstract} Minimizing the weight of an edge set satisfying parity constraints is a challenging branch of combinatorial optimization as witnessed by the binary hypergraph chapter of Alexander Schrijver's book ``Combinatorial Optimization" (Chapter 80). This area contains relevant graph theory  problems including open cases of the \textsc{Max Cut} problem and some multiflow problems. We clarify the interconnections between some of these problems and establish three levels of difficulties. 
		On the one hand, we prove that the \textsc{Shortest Odd Path} problem in undirected graphs without cycles of negative total weight and several related problems are $\NP$-hard, settling a long-standing open question asked by Lov\'asz (Open Problem~27 in Schrijver's book ``Combinatorial Optimization''). On the other hand, we provide an efficient algorithm to the closely related and well-studied \textsc{Minimum-weight Odd  $T$-Join} problem for non-negative weights: our algorithm runs in FPT time parameterized by~$c$, where $c$ is the number of connected components in some efficiently computed minimum-weight $T$-join.
  If negative weights are also allowed, then finding a minimum-weight odd  $\{s,t\}$-join is equivalent to the \textsc{Minimum-weight Odd $T$-Join} problem for arbitrary weights, whose complexity is still only conjectured to be polynomial-time solvable. 
		The analogous problems for digraphs are also considered. 
	\end{abstract}

	\section{Introduction}
	\label{sec:intro}
	
	The {\sc Minimum-weight Odd $T$-Join problem} (MOTJ) asks for an odd-cardinality $T$-join of minimum weight for a given subset~$T$ of vertices in an undirected, edge-weighted graph.
	The MOTJ problem is the graph special case of optimization in  binary hypergraphs. This area contains relevant  problems on graphs, including open cases of the {\sc Max Cut} problem, or some multiflow problems.  Stimulating minimax conjectures of Seymour's have been solved in special cases (see e.g., Guenin and Abdi~\cite{A,AG,Guenin}) but  optimization algorithms were not well-known,  not even for cases where the  minimax theorems conjectured by Seymour hold. 
	
	In this paper we study a handful of interconnected optimization problems  in undirected, edge-weighted graphs that  involve parity constraints on certain  edge sets. Such problems are considered in a general framework under the term of ``binary hypergraphs'' \cite[Chapter 80]{SYB}, subject to a huge number of deep results and conjectures since the seventies until now (see e.g.,~\cite{PDS,Geelen,Guenin1,Guenin,AG,A}; those published before 2002 are summarized in Schrijver's book~\cite{SYB}). A first round of problems like $T$-joins or odd cycles for non-negative edge-weights have been studied in the last century \cite{E,GP}. Then time has come for adding more binary constraints \cite{Guenin1,Guenin,Geelen}, bringing in new results and new challenges. 
	
	In this work we consider the main algorithmic challenges.  Among other variants, we study the problems of finding a minimum-weight odd $T$-join (MOTJ), a shortest odd cycle (SOC)  or a shortest odd path between two given vertices (SOP) in an undirected graph with \emph{conservative} weights, that is, 
	when negative weights but no  cycles with negative total weight %\andras{!}
 are allowed. Our results are the following:
	
	\begin{itemize}
		\item[(a)] The variant of SOC where the solution has to contain a given vertex of the graph is $\NP$-hard, {\em implying $\NP$-hardness for the SOP problem as well}. The complexity of the latter has been an open question by Lov\'asz (Open Problem~27 in Schrijver's  book~\cite{SYB}) for the last more than twenty years.
		
		\item[(b)] 
		The MOTJ problem for non-negative weights {\em can be solved in $2^{|T|/2} O(n^4)$ time} on an $n$-vertex graph. Our method is based on structural properties of shortest cycles in graphs with conservative weights, yielding an algorithm for SOC that is efficient when negative-weight edges span a bounded number of components.
		
		\item[(c)] The SOC problem for conservative weights is polynomially equivalent to MOTJ, and although we do solve certain special cases in polynomial time, {\em it remains open in general}.
	\end{itemize}

	We  prove in addition that finding two openly disjoint paths between two vertices with minimum total weight in an undirected graph with conservative weights in $\NP$-hard; this problem
	has, quite surprisingly, also been open.
	A major benefit of our results is finding connections among a so far chaotic set of influential problems, and sorting them into {\em polynomial-time solvable, open and $\NP$-hard cases} (cf.\ (a), (b), and (c) respectively). We will also see that some of the analogous problems for digraphs are easily reducible to tractable problems, while some others are equivalent with the problems we prove to be $\NP$-hard (see Sections~\ref{sec:NPC} and~\ref{sec:conn}); some new related open problems also arise (see Section~\ref{sub:conc}).
	
	\smallskip
	
	The  SOP problem  contains the following classical problem  SP:
	
	\medskip
	\noindent {\bf SHORTEST PATH IN UNDIRECTED CONSERVATIVE GRAPHS (SP)}
	
	\noindent {\bf Input}: An undirected  graph $G=(V,E)$ with weights  $w:E\rightarrow\mathbb{Z}$, $s, t\in V$, and $k\in \mathbb{Z}$.
	
	\noindent {\bf Question:} Is $G$ conservative with $w$, and if so, is there a path between $s$ and $t$ of weight at most~$k$?
	
	\medskip 
	Indeed, to solve SP using SOP, simply add a new vertex~$t'$, and add both an edge and a path of two edges from~$t$ to~$t'$, each with weight~0, then find a shortest odd \mbox{$(s,t')$-path}. As it is well-known,  SP can be solved in polynomial time, but this cannot be done via shortest path algorithms for conservative digraphs: 
	
	To solve SP by using techniques developed for digraphs, we would have to include each edge in both directions, and negative edges would lead to negative cycles consisting of two edges. Moreover, the algorithms for directed graphs are all based on the fact that subpaths of shortest paths are shortest and the triangle inequality holds, which is not true in  the undirected case. 
	In fact, SP contains the problem of finding a minimum-weight perfect matching,
	and conversely, deciding conservativeness or finding a shortest $(s,t)$-path with conservative weights -- which is the same as finding a shortest $\{s, t\}$-join with these weights -- can be easily reduced to weighted matching problems~\cite{E}, \cite[Section~29.2]{SYB}.  This establishes polynomial-time solvability for SP, but not as easily as  the analogous algorithms for directed graphs.

	\smallskip
	Requiring odd cardinality from the paths to be optimized on, will lead, as mentioned, to $\NP$-completeness.
	In fact, our $\NP$-hardness proof for SOP shows also the $\NP$-hardness of the \textsc{Shortest Odd Cycle through a Point} (SOCp) problem that asks for a shortest odd  cycle containing a given vertex in an undirected graph with conservative edge weights.

	\smallskip
	However, the \textsc{Shortest Odd Cycle} (SOC) problem of finding a shortest odd-cardinality cycle in undirected conservative graphs seems significantly easier. Although 
	SOC is known to be in~$\RP$ by a randomized polynomial algorithm due to Geelen and Kapadia~\cite{Geelen}, its polynomial-time solvability remains open.
	The SOC problem has been studied from multiple aspects and under various names; one of these is MOTJ. 
	The study of~SOC and MOTJ  has lead to deep structural results:

	Seymour \cite{PDS,SYB} conjectured minimax theorems for the problem of finding a shortest odd $T$-join  if certain minors are excluded; one of these, the {\em flowing conjecture} postulates the existence of a fractional dual solution for a minimum transversal of odd $T$-joins, while the {\em cycling conjecture} bets on the existence of an integer dual solution for non-negative weights.  We need not and will not enter these linear programming aspects in this note, but it is interesting to mention that these conjectures have been solved in the $|T|\le 2$ special case by  Guenin and Abdi~\cite{A,AG,Guenin}, without caring about algorithms.   
	On the other hand, a randomized polynomial algorithm  has been given for SOC by Geelen and Kapadia~\cite{Geelen}, 
	making polynomial-time solvability plausible.  In Section~\ref{sub:oddT-join} we discuss these connections. 
	
	\medskip
	\noindent {\bf Main contribution.}
	On the positive side, in Section~\ref{sec:Tractable} 
 we propose a fixed-parameter tractable (FPT) algorithm for the \textsc{Shortest Odd Cycle} problem for conservative weights,   parameterized by the number of connected components spanned by all negative edges (Theorem~\ref{thm:alg-fpt}). 
 As a consequence, MOTJ with non-negative weights can be solved by an algorithm that first computes a minimum-weight $T$-join~$F$, 
 and runs in time $2^c O(n^4)$ where $c$ is the number of connected components in~$F$, and $n$ the number of vertices in the graph.
 
 As a further corollary, 
	Cook, Espinoza and Goycoolea's FPT result~\cite{CEG} for MOTJ with parameter~$|T|$ and non-negative weights follows,  with a different proof and slightly different time complexity  (Corollary~\ref{cor:alg-fpt}).

        The main surprise -- causing at the same time some disappointment --  is {\em the $\NP$-completeness of SOCp   (Theorem~\ref{thm:SOP}), in contrast with SOC,  which remains open}. As an immediate corollary, surprisingly, Lov\'asz's problem SOP is also $\NP$-complete (Corollary~\ref{cor:SOP}). 
 In Corollary~\ref{cor:DISP} we further obtain $\NP$-completeness of the \textsc{Disjoint Shortest Paths} (DISP) problem,  where given two vertices, $s$ and~$t$ in an undirected, conservative graph, the task is to find two openly disjoint paths between~$s$ and~$t$ with minimum total weight.
	
	Finally, we present certain connections  including equivalences between the  studied but still open problems (Theorem~\ref{thm:equiv}).
	
	\medskip
	\noindent {\bf Organization.}
	In Section~\ref{sec:prelim} we introduce the most important notations, terminology and some basic facts. 
	In Section~\ref{sec:Tractable} we make an inventory of the positive results concerning MOTJ. Besides  mentioning  some existing results and recalling the main conjecture about MOTJ, 
	simple  structural results are presented for the \textsc{Shortest Odd Cycle} problem in conservative graphs,  leading first to a polynomial algorithm  for MOTJ with non-negative weights in the case when we can find a minimum-weight $T$-join that is connected (Section~\ref{sub:algoT2}), and then to efficient algorithms for MOTJ and for SOC (Section~\ref{sub:fix}).
 Our $\NP$-hardness results for the problems SOCp, SOP, and DISP are  presented in Section~\ref{sec:NPC}.
	
	The results of the paper reveal new possibilities for considering  special cases that may deserve more focus. We put forward  their relations and some open problems concerning them,  summarized in Section~\ref{sec:conn}, together with some conclusions. 
	
	\section{Preliminaries}
	\label{sec:prelim}
	We start with basic notation for graph-theoretic concepts and for properties of edge-weight functions. We then proceed by giving the precise definitions of the problems already mentioned in Section~\ref{sec:intro} and stating some well-known facts about them that will be useful later on.

	\medskip
	\noindent
	{\bf Notation for graphs.}
	Given an undirected graph $G=(V,E)$,  for some $F \subseteq E$ and $v \in V$ let $d_F(v)$ denote the {\em degree} of~$v$ in~$F$, i.e., the  number of edges in~$F$ incident to~$v$.
	Let $V(F)$   denote the set of vertices that are incident to some edge in~$F$. Let $G[F]$ denote the subgraph of~$G$ \emph{spanned} by~$F$, that is, the graph  $(V(F),F)$. 
	
	A {\em cycle} in an undirected graph $G=(V,E)$ is a nonempty set~$C$ of edges such that $G[C]$ is connected, and $d_C(v)=2$ for each vertex~$v \in V(C)$. 
	In a digraph $G=(V,E)$, a (directed) {\em cycle} additionally satisfies that all in- and out-degrees in $G[C]$ are equal to~$1$. 
	For two distinct vertices~$s$ and~$t$ in a graph, an \emph{$(s,t)$-path} has the same definition except that 
	the two \emph{endpoints}, $s$ and~$t$,  have degree~$1$ in the undirected case; 
	in the directed case,  $s$ has in-degree~$0$ and out-degree~$1$, while $t$ has out-degree~$0$, and in-degree~$1$. 
	A cycle $C$ with $s\in V(C)$ is also considered to be an $(s,t)$-path with $s=t$. 
	For two sets of vertices $S, T\subseteq V$ of a graph, an {\em $(S,T)$-path} is an $(s,t)$-path for some $s \in S$ and $t \in T$. If $P$ is a path and $a,b\in V(P)$, then the subpath of $P$ between $a,b\in V(P)$ is denoted by $P(a,b)$. 
	
	Note that we have defined cycles and $(s,t)$-paths as edge sets. 
	With a slight abuse of terminology, a
	\emph{path} in~$G$ may also be  a subgraph spanned by an $(s,t)$-path for distinct vertices~$s$ and~$t$, and we also consider a single vertex in~$G$ to be a \emph{trivial} path.
	
	Two paths are said to be {\em vertex-disjoint} (or {\em edge-disjoint}) if they do not have a common vertex (or edge, respectively), and they are said to be {\em openly disjoint} if they can only share vertices that are endpoints on both paths.
	
	A \emph{$T$-join} in an undirected graph $G=(V,E)$ for some $T\subseteq V$  is a subset of edges $J\subseteq E$, such that $d_J(v)$ is odd if $v\in T$, and even if $v\in V\setminus T$. An $\emptyset$-join is the disjoint union of cycles; inclusionwise minimal, non-empty $\emptyset$-joins are exactly the cycles. A $T$-join with  $|T|=2$, that is, with $T=\{s,t\}\subseteq V$ is the disjoint union of an $(s,t)$-path and some cycles, so the inclusionwise minimal ones are  $(s,t)$-paths.

	A cycle, a path, a $T$-join and  generally, any edge set is {\em odd (even)} if it contains an odd (respectively, even) number of edges.

	\medskip
	\noindent
	{\bf Weight functions.}
	For a function $f:D\rightarrow \mathbb{R}$ and some~$D' \subseteq D$, let $f(D')\leteq \sum_{d \in D'} f(d)$, as usually. 
	In an optimization problem over a set of feasible edge sets (e.g., over paths between two vertices, or over all cycles),
	a \emph{$w$-minimum} solution is one that  has minimum weight according to a given edge-weight function~$w$; if $w$ is clear from the context, we might also say that such a solution is \emph{shortest}.
	
	An undirected graph~$G=(V,E)$ is {\em conservative}  with weights $w:E\rightarrow \mathbb{R}$, if  $w(C)\ge 0$ for any cycle~$C$ of~$G$.

	\smallskip
	
	For arbitrary $w:E\rightarrow \mathbb{Z}$ and  $F\subseteq E$, let  $w[F]: E\rightarrow \mathbb{Z}$ denote the function defined by 
	\[
	w[F](e)\leteq \left\{ \begin{array}{rl}
		-w(e) & \textrm{ if $e\in F$,} \\
		w(e) & \textrm{ if $e \in E \setminus F$.}
	\end{array} \right.
	\]

	Denote the \emph{symmetric difference} of two sets $X$ and $Y$ by $X\Delta Y:=(X\setminus Y)\cup (Y\setminus X)$.
	Then clearly $w[F](X)= w(X\setminus F)- w(X\cap F)=w(X\Delta F) - w(F)$ for any $X\subseteq E$. In particular, {\em $F$ is a $w$-minimum $T$-join for some vertex set~$T$, if and only if $w[F]$ is conservative} we will refer to this as {\em Guan's Lemma} (stated by Guan~\cite{G} for  the ``Chinese Postman problem'' for non-negative weights). 
	Indeed, for any cycle $X$, the set~$X\Delta F$ is also a $T$-join, therefore $w[F](X)=w(X\Delta F) - w(F)\ge 0$ by the definition of~$F$. 
	
	\medskip
	For simplicity we can and will often  suppose that $w(e)\ne 0$ for all $e\in E$, to avoid $0$-weight cycles. 
	In fact, it will be convenient to assume that the weight function is \emph{normal}, meaning that no  edge 
	has $0$ weight, no cycle has $0$ weight,
    and edge sets of different cardinality have different weights. 
	Normality will be supposed in some proofs where it does not restrict generality; it is helpful for avoiding some technical detours.
	
	We can \emph{normalize} a given conservative, rational weight function by first multiplying by the smallest common denominator, which does not change the optimal sets, increases the size of the input only polynomially, and can be carried out in polynomial time.  An integer  weight function~$w:E \rightarrow \mathbb{Z}$ is normalized by  defining $w'(e)\leteq |E|\cdot w(e)+1$ for each edge~$e \in E$. The normalized weight function~$w'$ will satisfy $w'(X) < w'(Y)$ for each pair of edge sets~$X$ and~$Y$ with $w(X) < w(Y)$, and equal-weight edge sets of different cardinality will get different $w'$-weights. 
	When searching for a minimum-weight edge set with a given property, normalization does not essentially change the problem, since at least one optimal edge set will remain optimal. 
	Note that if~$w$ is conservative, then $w'$ will also be conservative, furthermore, {\em odd sets will have different weights from even ones.}
	In particular the only $\emptyset$-join of weight~$0$ will be the empty set.
	
	\medskip
	\noindent
	{\bf Problem definitions and classical results.} Consider the following two  problems, which differ only in that the second one confines the searched cycle to contain a given vertex:   
	
	\medskip 
	\noindent {\bf SHORTEST ODD CYCLE IN CONSERVATIVE GRAPHS (SOC)}
	
	\noindent {\bf Input}: An undirected graph $G=(V,E)$ conservative with $w:E\rightarrow\mathbb{Z}$   and $k\in\mathbb{Z}$.
	
	\noindent {\bf Question:} Is there an odd cycle $C$ in~$G$ whose weight is at most $k$?
	
	\medskip 
	\noindent {\bf SHORTEST ODD CYCLE IN CONSERVATIVE GRAPHS THROUGH A POINT (SOCp)}

	\noindent {\bf Input}: An undirected graph $G=(V,E)$ conservative with $w:E\rightarrow\mathbb{Z}$ ,  $p\in V$, and $k\in\mathbb{Z}$.

	\noindent {\bf Question:} Is there an odd cycle $C$ in~$G$ with $p\in V(C)$ whose weight is at most~$k$?

	\medskip
	The following problem MOTJ is closely related to SOC: on the one hand, MOTJ is a generalization of SOC (consider the case $T=\emptyset$, which yields exactly SOC), on the other hand SOC is exactly the problem of finding the ``improving step'' for reaching an optimum in MOTJ; the two problems are therefore polynomially equivalent (cf.~Theorem~\ref{thm:equiv}). 
	
	\medskip
	\noindent {\bf MINIMUM-WEIGHT ODD $T$-JOIN (MOTJ)}
	
	\noindent {\bf Input}: An undirected  graph $G=(V,E)$ with $w:E\rightarrow\mathbb{Z}$, $T\subseteq V$, and $k\in\mathbb{Z}$.

	\noindent {\bf Question:} Is there an odd $T$-join in~$G$ with total weight at most $k$?

	\medskip
	\noindent {\bf SHORTEST ODD PATH IN CONSERVATIVE GRAPHS (SOP)}
	
	\noindent {\bf Input}: An undirected  graph $G=(V,E)$ conservative with $w:E\rightarrow\mathbb{Z}$, $s, t\in V$, and $k\in\mathbb{Z}$.
	
	\noindent {\bf Question:} Is there an odd  $(s,t)$-path in~$G$ with total weight at most $k$?

	\medskip
	Analogously to SOCp and SOP, we also define the problems SECp and SEC by replacing ``odd'' with ``even'' in the definitions.
	
	The problems that will turn out to be $\NP$-hard (SOCp and SOP) will actually already be  $\NP$-hard for weight functions with values in $\{-1,1\}$. We denote the problems restricted to such weight functions by putting $\pm 1$ in subscript, for example SOP$_{\pm 1}$ means SOP restricted to weight functions taking only values  from~$\{-1,1\}$. A subscript $+$ means a restriction to non-negative weights. 
	The following theorem summarizes well-known results:
	
	\begin{theorem}\label{thm:classical}
		SP,  SOP$_+$, SEP$_+$, SOCp$_+$,  SECp$_+$, and SOC$_+$, SEC$_+$  are polynomially solvable. 
	\end{theorem}
	
	\begin{proof} We saw in the introduction that SP can be solved in polynomial time (see \cite[Section~29.2]{SYB}). 
		SOP$_+$ and SEP$_+$ can be solved in polynomial time by the well-known  ``Waterloo folklore'' algorithm related to Edmonds' classical work on matchings~\cite{GP}, see also \cite[Section 29.11e]{SYB}. Then SOCp$_+$ on an instance $(G=(V,E),w,p,k)$ can be solved  by solving SEP$_+$ on~$(G'=(V,E\setminus\{pr\}),w',p,r,k')$ for each edge $pr\in E$ incident to $p$, where $w'$ is the restriction of $w$ to $E\setminus\{pr\}$ and $k'=k-w(pr)$. 
		We can reduce SECp$_+$ to SOP$_+$ similarly. 
		Finally, SOC$_+$ on an instance $(G=(V,E),w,k)$ can be solved by computing SOCp$_+$ on $(G,w,p,k)$ for all~$p\in V$, and SEC$_+$ can be reduced to SECp$_+$ similarly.   The execution time of all these problems is polynomial in the input size.
	\end{proof}
	
	We will see that SOC$_+$ and SEC$_+$ are actually much easier than matchings: they can be solved by using only Dijkstra's shortest path algorithm (see Proposition~\ref{prop:non-neg}). Problems concerning odd or even paths are not really different, since they can be reduced to one another by introducing a new vertex~$t'$ and an edge $tt'$. However, no such reduction is known between problems concerning odd and even cycles. In fact, even the existence of non-empty even cycles happens to be inherently more difficult, to the extent that its complexity is not yet completely settled in directed graphs; see  Proposition~\ref{prop:dir-hardness}, and in Section~\ref{sub:classic}.

	We finish the list of helpful preliminaries with further notations and observations:

	Given a graph $G= (V, E)$ and a  conservative weight function $w: E \rightarrow \mathbb{Z}$, we denote the set of edges with negative weight by~$E^-=\{e\in E: w(e)<0\}$, and let us write  $E^+ = E \setminus E^-$.
	Observe that each connected component~$K$ of~$G[E^-]$ is a tree, because $w$ is conservative on~$G$.  For any two vertices $u$ and $v$ in $K$, let $K(u,v)$ denote the unique $(u,v)$-path in $K$.  	 
	
	\begin{proposition}\label{prop:pathonepath}
		Suppose $G=(V,E)$ is conservative with $w$, and $P$ is a $w$-minimum $(u,v)$-path for some vertices $u,v\in V$. Then for each connected component~$K$ of~$G[E^-]$, either~$P$ and~$K$ are vertex-disjoint, or their intersection is a path. 
	\end{proposition}
	
	\begin{proof}
		For a contradiction, suppose that there is a connected component~$K$ of~$G[E^-]$ whose intersection with~$P$ is non-empty, and not a path. Then there exist two distinct vertices~$a$ and~$b$ in~$V(P)\cap V(K)$ so that $K(a,b)$, is edge-disjoint from $P$.
		
		Using that $w$ is conservative on~$G$, we get
		$w(P(a,b) \cup K(a,b)) \geq 0$. Since every edge in~$K$ has negative weight, this implies $w(K(a,b))< 0 < w(P(a,b))$. Then $w(P\setminus P(a,b) \cup K(a,b)) < w(P)$, contradicting the choice of $P$. 
	\end{proof}
	
	\section{Are MOTJ and SOC tractable? }\label{sec:Tractable}

	In this section we collect evidence that MOTJ is tractable in its full generality, and present some  new cases when this can be already proved. 
	
	The conjecture of polynomial-time solvability of MOTJ is first of all supported by Geelen and Kapadia's result \cite{Geelen} establishing that MOTJ belongs to $\RP$,  saving the problem from being suspected to be NP-hard (which would imply $\NP=\RP$), and suggesting the following conjecture:  
	
	\begin{conjecture} 
		\label{conj:SOC}
		MOTJ and SOC can be solved in polynomial time. 
	\end{conjecture}  
	
	This conjecture is equivalent with a whole range of equivalent conjectures, since MOTJ can be reduced to several special cases, including the case when weights are non-negative, or when~$|T| \leq 2$  (see Theorem~\ref{thm:equiv} and some remarks thereafter). 
	However, restricting~$|T|$ and simultaneously assuming non-negative weights seems to make the problem easier (Corollary~\ref{cor:alg-fpt}), 
	confirming Conjecture~\ref{conj:SOC} under these assumptions. In Section~\ref{sub:algoT2} we present two approaches for solving MOTJ$_+$ for $|T|=2$, and the second one is generalizable to the case when a minimum $T$-join with a bounded number of components is given.  In Section~\ref{sub:fix} we investigate the general case that leads us to an FPT algorithm for the equivalent SOC where the parameter is the  number of connected components spanned by all negative edges.  
 As a corollary we also obtain the fix parameter tractability of  MOTJ$_+$ with parameter~$|T|$, already proved by   Cook, Espinoza and Goycoolea~\cite{CEG}, with a slightly worse  dependence on~$|T|$, but better on the number or vertices.

	\subsection{\texorpdfstring{MOTJ$_+$ with a connected minimum-weight $T$-join}{MOTJ+ with a connected minimum-weight T-join}}
	\label{sub:algoT2}
	
	We start by proving that MOTJ$_+$ is polynomial-time solvable if $|T|\le 2$, a case for which Seymour's conjectures mentioned in the introduction have also been proved.
	(Guenin~\cite{Guenin}  characterized for $|T| \leq 2$, in terms of the two small excluded minors of Seymour, when inclusionwise minimal  odd $T$-joins are ``ideal'';   Abdi and Guenin \cite{A,AG}  proved that  in this special case actually a stronger minimax theorem holds.) 

    A simple $O(n^3)$ algorithm is known in this case from Cook, Espinoza, Goycoolea \cite[Proposition 5.3]{CEG} (see after Corollary~\ref{cor:alg-fpt}). Our goal here is to introduce the reader to certain structural properties of shortest odd cycles in conservative graphs that also allow us to solve MOTJ$_+$ when a minimum-weight $T$-join is connected.  These properties will then carry us further, to the case of several components.
	
	We first state a clarifying observation on inclusionwise minimal odd $T$-joins for $|T|=2$:

	\begin{lemma}[\!\cite{AG,Guenin}] 
		\label{lem:AG}
		Let $G=(V,E)$ be a graph,  $s, t\in V$, and $F\subseteq E$. Then $F$ is an  inclusionwise minimal odd $\{s,t\}$-join, if and only if it is an odd $(s,t)$-path  or it is of the form $P\cup C$ where $P$ is an even  $(s,t)$-path and $C$ an odd cycle that is edge-disjoint from~$P$ and satisfies $|V(P)\cap V(C)|\le 1$. 
	\end{lemma} 
	
	\begin{proof} Clearly, any odd $(s,t)$-path  and any edge set $P \cup C$ as defined in the statement of the lemma is an inclusionwise minimal odd $\{s,t\}$-join. Conversely, any $\{s,t\}$-join~$F$ is the union of an $(s,t)$-path $P$ and pairwise edge-disjoint cycles. So if $F$ is an inclusionwise minimal odd $\{s,t\}$-join, then it   contains neither even cycles, nor more than one odd cycle. (An even cycle or two odd cycles could be deleted from~$F$, contradicting  minimality.)
		
		If $F$ contains no cycle, then $F=P$ where $P$ is an odd $(s,t)$-path. Otherwise, $F=P\cup C$ where $P$ and $C$ {\em are edge-disjoint}, and $C$ is an odd cycle; $P$ is then an even path, since $|F|$ is odd.  It remains to prove that $|V(P)\cap V(C)|\le 1$. 
		
		Suppose for a contradiction $|V(P)\cap V(C)| > 1$. Traversing ~$P$ from $s$ to~$t$, let $a$ and~$b$ be the first and say the last encountered vertex of $C$, respectively.  Since $|V(P)\cap V(C)| \ge 2$ we have that both $a$, $b$ exist, $a\ne b$, and therefore $a$ and $b$ divide $C$ into two $(a,b)$-paths, $C_1$ and $C_2$. So $P\cup C$ contains  three, pairwise edge-disjoint $(a,b)$-paths:  $P(a,b)$, $C_1$, and $C_2$,  two of which necessarily have the same parity. Deleting those two from $F$ we still get an odd  $\{s,t\}$-join, contradicting the  inclusionwise minimality of $F$.  
	\end{proof}
	
	\smallskip
	\noindent
	{\bf A simple algorithm for MOTJ$_+$ with $|T|=2$.}
	Lemma~\ref{lem:AG} easily yields a simple polynomial algorithm for finding  a minimum-weight odd (or even)  $\{s,t\}$-join $F$. For normalized weights (it is sufficient not to have $0$-weight cycles) such an $F$ is clearly inclusionwise minimal, and thus can be searched in the form given in Lemma~\ref{lem:AG}:
	\begin{itemize}[leftmargin=38pt,itemindent=16pt,parsep=6pt] 
		\item[{\bf Step 1.}] Compute in the input graph a minimum-weight odd  $(s,t)$-path $P_{\textup{odd}}$, a minimum-weight even $(s,t)$-path $P_{\textup{even}}$, and a minimum-weight odd cycle~$C$.
		\item[{\bf Step 2.}] Let $F$ be the shorter one among $P_{\textup{odd}}$ and $P_{\textup{even}} \Delta C$; if their weights are equal, choose arbitrarily.
	\end{itemize} 	
	The correctness of the above algorithm follows from Lemma~\ref{lem:AG}. To see this, it suffices to observe that  $P_{\textup{even}} \Delta C$ is an odd $\{s,t\}$-join, and its weight does not exceed that of any $P'\cup C'$ where $P'$ is an even path and $C'$ an odd cycle edge-disjoint from~$P'$, since $w(P) \le w(P')$, and $w(C)\le  w(C')$. 
	Therefore, by Lemma~\ref{lem:AG} we know that $F$ indeed has minimum weight among all inclusionwise minimal $\{s,t\}$-joins.  
	Furthermore, according to Lemma~\ref{thm:classical}, each of~$P_{\textup{even}}$, $P_{\textup{odd}}$, and~$C$ can be computed in polynomial time. 
	
	\smallskip

	We also introduce another approach for the case $|T| = 2$ that brings us closer to the extension of 
	polynomial solvability of MOTJ$_+$ when a minimum $T$-join of a constant number of components can be constructed.
	Our next algorithm relies heavily on Proposition~\ref{prop:reductocycle}, illustrated by Figure~\ref{fig:reductocycle}. 
	Recall that inclusionwise minimal odd $\emptyset$-joins are cycles, and recall also Lemma~\ref{lem:AG} and Guan's Lemma: 
	\begin{figure}[t]
		\centering
		\includegraphics[scale=1]{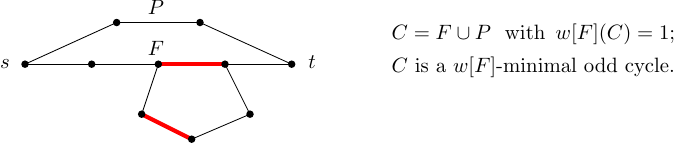}
		\caption{Illustration of Proposition~\ref{prop:reductocycle}. Bold, red lines depict edges of weight~$-1$,  and all the other edges have weight~1. The shortest $(s,t)$-path $F$ is the middle horizontal line; it is even and of weight~$2$. 
			According to Proposition~\ref{prop:reductocycle},  $P=F\Delta C$ is a $w$-shortest odd $\{s,t\}$-join.}
	\label{fig:reductocycle}
\end{figure}

\begin{proposition}\label{prop:reductocycle} 
	Let $w:E\rightarrow \mathbb{R}$ be  arbitrary, and $T\subseteq V$ with $|T|$ even. If 
	$F$ is a $w$-minimum $T$-join and $|F|$ is even, then $F_{\rm odd}$ is a $w$-minimum odd $T$-join if and only if $F_{\rm odd}= F\Delta C$ for some $w[F]$-minimum  odd cycle  $C$.
\end{proposition}

\begin{proof}
	Let $F$ be as in the condition. Then correspondence $F_{\rm odd} \leftrightarrow F_{\rm odd} \Delta F$ is a bijection between odd $T$-joins and odd $\emptyset$-joins. 
	Note further that 
	\[w[F](F \Delta F_{\rm odd})=
	w(F_{\rm odd}\setminus F) - w(F\setminus F_{\rm odd})=w(F_{\rm odd}) - w(F).\]
	Since $w(F)$ is a fixed value, we obtain that $F \Delta F_{\rm odd}$ minimizes $w[F]$ over all odd $\emptyset$-joins exactly if $F_{\rm odd}$ minimizes $w$ over all odd $T$-joins.
	It remains to observe that since $w[F]$ is conservative, a $w[F]$-minimal  odd $\emptyset$-join must consist of a $w[F]$-minimal odd cycle and possibly some additional cycles whose $w[F]$-weight is zero.
\end{proof}

By Proposition~\ref{prop:reductocycle}, finding a $w$-minimal odd $T$-join in an instance $(G,w,T,k)$ of MOTJ can be reduced to finding a $w[J]$-minimum odd cycle in the same graph~$G$ where $J$ is a $w$-minimum $T$-join.
Next we present Proposition~\ref{prop:pathonepath} for shortest odd cycles:
\begin{lemma}\label{lem:onepathgen} Suppose the graph $G=(V,E)$ is conservative with $w:E \rightarrow \mathbb{R}$, and $C$ is a $w$-minimum odd cycle. Then for each connected component~$K$ of~$G[E^-]$, either $C$ and~$K$ are vertex-disjoint, or their intersection is  a path. 
\end{lemma} 
\begin{proof}
	The proof is essentially the same as that of Proposition~\ref{prop:pathonepath}.
	For a contradiction, suppose that there is a connected component~$K$ of~$G[E^-]$ whose intersection with~$C$ is not a path. Then there exist two distinct vertices~$a$ and~$b$ in~$V(C)\cap V(K)$ so that the unique $(a,b)$-path in~$K$, denoted by $K(a,b)$, is edge-disjoint from $C$, and consequently, from both $(a,b)$-paths $C_1$ and $C_2$ into which $a$ and $b$ divides $C$; note that $|C_1|$ and $|C_2|$ have different parity.
	
	Using that $w$ is conservative on~$G$, we get
	$w(C_i \cup K(a,b)) \geq 0$ for both $i=1,2$. Recall also that $w(K(a,b))< 0 $ since every edge in~$K$ has negative $w$-weight.  
	Hence, $w(C_i \cup K(a,b))<w(C_i)$ and $w(C_i)>0$ for $i=1,2$. Therefore $C_i \cup K(a,b)$ for $i=1, 2$ are two cycles of weight less than~$w(C)$, and one of them is odd, a contradiction to the definition of $C$. 
\end{proof}

\begin{corollary}\label{cor:onepath} Let $G=(V,E)$ be a graph with weight function $w:E \rightarrow{R}_0^+$, let $P$ be a $w$-minimum $(s,t)$-path for some vertices $s,t\in V$, and let $C$ be a $w[P]$-minimum odd cycle.  Then either $C$ and $P$ are vertex-disjoint, or their intersection is a path.  
\end{corollary} 

\begin{proof} Since $w$ is non-negative, a $w$-minimum $(s,t)$-path is also a $w$-minimum $\{s,t\}$-join. By 
	Guan's Lemma mentioned in the introduction, $w[P]$ is therefore conservative. Hence, the statement follows directly from Lemma~\ref{lem:onepathgen}.
\end{proof}

\medskip
\noindent
{\bf Another simple algorithm for MOTJ$_+$ with $|T|=2$,} i.e., a second polynomial algorithm  for finding a minimum weight odd $\{s,t\}$-join for a non-negative weight function $w$, using Proposition~\ref{prop:reductocycle}
and Corollary~\ref{cor:onepath}. 

\begin{itemize}[leftmargin=38pt,itemindent=16pt,parsep=6pt] 
	\item[{\bf Step 1.}]
	Compute a $w$-minimum $(s,t)$-path~$J$. If $J$ is odd, return it and stop.
	\item[{\bf Step 2.}]
	Let $C$ be the cycle of smaller $w[J]$-value between
	\begin{itemize}[topsep=0pt]
		\item[(a)] a $w$-shortest odd cycle in $E \setminus J$, 
		\item[(b)] a cycle $J(u,v)\cup P$ with the smallest $w[J]$-weight 
		where $u, v\in V(J)$, $u \ne v$, and $P$ is a $w$-minimum $(u,v)$-path in $E \setminus J$ of parity different from $J(u,v)$. 
	\end{itemize}
	\item[{\bf Step 3.}] Return $J \Delta C$.
\end{itemize}

Computing a $w$-minimum $\{s,t\}$-path in Step~1
can be done in polynomial time by Theorem~\ref{thm:classical}. (In fact, we only need SP$_+$ here, solved e.g., by Dijkstra's algorithm.) 
 Step~2(a) is a SOP$_+$ problem which is solvable in polynomial time by Theorem~\ref{thm:classical}, and so are the SEP$_+$ and SOP$_+$ problems of Step~2(b).  Hence, the above algorithm runs in polynomial time.
 
 Since $w$ is non-negative, the $w$-minimum $(s,t)$-path~$P$ of Step~1 is a $w$-minimum $\{s,t\}$-join, so if it is odd, then the algorithm returns the correct result. 
It remains to prove that {\em the odd cycle~$C$ computed in Step~2 is a $w[J]$-minimum odd cycle}, since then the correctness of the algorithm follows immediately from
Proposition~\ref{prop:reductocycle}. 

Let $C'$ be a $w[J]$-minimum odd cycle.
The edges with negative $w[J]$-weight are exactly the edges of~$J$, and 
so by Corollary~\ref{cor:onepath}, either $C'$ is edge-disjoint from~$J$ (implying that $C'$ is a $w$-shortest cycle in $E \setminus J$), or $C'$ is formed by the union of a non-trivial subpath of $J$ and a path in $E \setminus J$ of different parity between the same vertices. Hence, we obtain an odd cycle with $w[J]$-weight at most~$w[J](C')$ 
either in Step~2(a) or in Step~2(b).

\medskip
It is not hard to see that the same approach works for solving MOTJ$_+$ for the more general case when a connected $w$-minimum $T$-join  is given.  
In the next section, we generalize this algorithm to work without any assumption on the number of components.
The running time of the equivalent SOC, however, is fixed-parameter tractable if the parameter is the number of  connected components of negative edges. 

\subsection{FPT algorithms for SOC and MOTJ\texorpdfstring{$_+$}{+}}
\label{sub:fix}

In this section we show that MOTJ$_+$ is polynomial-time solvable when we are given a  $w$-minimum $T$-joins with a fixed number of connected components. By Proposition~\ref{prop:reductocycle}, this is reduced to 
finding a shortest odd cycle with respect to a conservative weight function with a fixed number of \emph{negative components}, defined as the connected components of~$G[E^-]$, the subgraph spanned by all negative-weight edges. 
The main result of this section is Theorem~\ref{thm:alg-fpt} establishing  that SOC is fixed-parameter tractable with parameter $c$, the number of negative components of the input graph. 
Conjecture~\ref{conj:SOC} remains open.

Finding a shortest odd cycle when there is only one negative component is based on  Lemma~\ref{lem:onepathgen}  via Corollary~\ref{cor:onepath}. 
Even though the assertion in Lemma~\ref{lem:onepathgen} holds for an arbitrary number of components, a naive, ``brute force'' approach based on this lemma provides only an algorithm with $c^cO(n^{2(c+1)})$ running time where $n=|V|$, which is polynomial if $c$ is fixed, but does not confirm fixed-parameter tractability. 
The  structural observations of  Lemma~\ref{lem:notshortest} below make it possible to achieve a running time  of  $2^c O(n^4)$. 

To state Lemma~\ref{lem:notshortest} which, together with the crucial property expressed in Lemma~\ref{lem:onepathgen} will form the basis of our FPT algorithm with parameter~$c$,
we need some additional notation.
We will use the notation $[k]:=\{1, \ldots, k\}$ for $k\in\mathbb{N}$.
Let $K_1, \dots, K_c$ be the edge sets of the negative  components in~$G$.
We define the graph~$G_I$ as $G_I=(V,E_I)$ where $E_I=E^+ \cup \bigcup_{i \in I} K_i$ for any  index set~$I \subseteq [c]$; in particular, $G_\emptyset=(V,E^+)$ and $G_{[c]}=G$.
For $F\subseteq E$ let $I(F):=\{i\in [c]:  F\cap K_i\ne\emptyset\}$ denote the set of indices of negative components having an edge in $F$. 
Whenever we consider any subgraph of~$G$ with weight function~$w$, we will implicitly use the restriction of~$w$ to this subgraph; this will not cause any confusion.

\begin{lemma}\label{lem:notshortest} 
Suppose $G=(V,E)$ is conservative with $w$, $w$ is normal, and   $C\subseteq E$ is a $w$-minimum odd cycle 
partitioned into 
paths $P_1,\ldots, P_m\subseteq E$  ($m\in\mathbb{N}, m\ge 2$) so that the sets~$I(P_i)$, $i\in [m]$, are pairwise disjoint. Then for any family of pairwise disjoint sets $I_i$, $i \in m$, with $I(P_i) \subseteq I_i\subseteq [c]$, statements~(a) and (b) hold for  $G_i:=G_{I_i}$:  
	\begin{itemize}
		
		\item[(a)]  Each $P_i$ for $i\in[m]$ but at most one is shortest among paths in~$G_i$  between the endpoints of~$P_i$, where the exception is also  shortest in~$G_i$ among paths of the same parity as~$P_i$. 
		\item[(b)] Suppose that the path $P_i'$ for each $i \in [m]$ is  shortest among paths in~$G_i$ of the same parity as~$P_i$ between the endpoints of~$P_i$. Then the paths $P_i'$, $i \in [m]$, are pairwise openly disjoint.
	\end{itemize}
	Furthermore, there exists a  partition of~$C$ into $P_1,\ldots, P_m$ and index sets $I_1, \ldots, I_m$ satisfying the above conditions with $m\le 3$ and such that either $P_1=\{e\}$ for some~$ e\in E_+$ or $P_1=C\cap K_j$ for some~$j\in[c]$, and if $m=3$, then $I_1=I(P_1)$, $I_2=I(P_2)$, $I_3=[c]\setminus (I_1\cup I_2)$,   	and~$P_i$ is shortest in~$G_i$ for~$i \in [3]$.  
\end{lemma}

	We call such a partition into two or three paths a \emph{compact} partition.  
	In a compact partition the only path~$P_i$ that may not be shortest in~$G_i$ between its endpoints is the path~$P_1$, and only if $m=2$ and $P_1=\{e\}$ for some~$e \in E^+$,
 when an even path shorter than~$w(e)$ may exist in~$G[E^+]$.
	
	Finally, note that the condition on the disjointness of index sets is a formalization of the requirement that each path of the partition should have ``its own negative components'' that are not used by any of the other paths.  We know from Lemma~\ref{lem:onepathgen} that the intersections with these components are paths. 

\begin{proof} 
	To see~(a), suppose first for a contradiction that there is a path~$P_i'$ in~$G_i$ of the same parity as~$P_i$ and between the same endpoints with weight $w(P_i')<w(P_i)$. Then $C'= (C \setminus P_i) \Delta P_i'$ is an odd $\emptyset$-join.  
	Since $P_i'\subseteq E_{I_i}$, our assumption on the disjointness of~$I_i$ from any $I_j$ where $j \neq i$ implies that $C\setminus P_i$ and~$P_i'$ may share only edges in~$E^+$. 
	As a consequence, $C'$ has weight at most~$w(C)-w(P_i)+w(P'_i)<w(C)$, a contradiction to the definition of $C$. 
	
	To finish the proof of (a), 
 suppose that there exist two distinct indices $i$ and $j$ in $[m]$ so that  $P_i$ and~$P_j$ are not shortest paths between their endpoints in~$G_i$ and in~$G_j$, respectively. Then let $\hat{P}_i$ and~$\hat{P}_j$ be shortest paths between the endpoints of~$P_i$ and $P_j$, respectively, in $G_i$ and~$G_j$.
	We conclude that  the parity of $P_i$ and $P_j$ differs from the parity of $\hat{P}_i$ and $\hat{P}_j$, respectively, since the former two are shortest for their parity but not shortest, while the latter two are shortest.
	Using again our assumptions on $I_i$ and~$I_j$,  we obtain that
	$\hat{C}= C \setminus (P_i \cup P_j) \Delta \hat{P}_i \Delta \hat{P}_j$ is an odd $\emptyset$-join with weight at most~$w(C)-w(P_i)-w(P_j)+w(\hat{P}_i)+w(\hat{P}_j)<w(C)$, a contradiction. 
	
	In order to prove (b), note first that the paths $P_i'$, $i \in [m]$,  must be {\em pairwise edge-disjoint}, since, using similar arguments as  before, we know that any two of them can only share edges of~$E^+$, and if they do have a common edge, then there exists a smaller cycle; more formally: 
 let $i, j\in[m]$ with $i \neq j$, and define  $C' := (C\setminus (P_i\cup P_j) \Delta P_i'\Delta P_j'$, which  is an odd $\emptyset$-join with $w(C')\le w(C)$, so by the minimality of $C$ the equality holds here.   Since  $I(P_i)\cap I(P_j)=\emptyset$, a common edge of $P_i$ and $P_j$ would have positive weight (since $w$ is normal) which would imply $w(C')<w(C)$. Thus, $P_i$ and $P_j$ are edge-disjoint. 
	
Suppose now for a contradiction that 
	$P_i'$ and~$P_j'$ ($i, j \in [m], i \neq j$) are not openly disjoint, so there exists some
	$x\in V(P_i')\cap V(P_j')$ that is an inner vertex of at least one of $P_i'$ and $P_j'$. Then $C'$, defined as above, contains a cycle $C''$ as a non-empty proper subset, so  
	$C''$  and $C'\setminus C''$ partition~$C'$ into two non-empty $\emptyset$-joins, exactly one of which is odd, denote it by~$Q$. Since there are no cycles of weight~$0$ in~$G$ by the normality of~$w$, we get that $w(Q)< w(C')=w(C)$, a contradiction.
	
	Finally, in order to prove the last sentence of the lemma,  choose first $P_2\subseteq C$ so that it satisfies $I(P_2)\cap I(C\setminus P_2)=\emptyset$ and is a shortest path in $G_{I(P_2)}$ between two distinct vertices $u$ and~$v$ on~$C$. 
	To see that such a path exists,  consider any partition of~$C$ into paths~$Q_1$ and~$Q_2$ with $I(Q_1)\cap I(Q_2)=\emptyset$; then by (a) at least one~$Q_i$, $i \in [2]$, is shortest in $G_{I(Q_i)}$.  Moreover, choose~$P_2$ so that it is inclusionwise maximal among all paths satisfying these requirements.
	
	Choose $P_1 \subseteq C \setminus P_2$ so that it is consecutive with~$P_2$ on~$C$, and consists either of one positive edge, or of a negative path that continues until the next positive edge on~$C$. In the latter case, by Lemma~\ref{lem:onepathgen} 
	this path contains the entire intersection of $C$ with the component of $G[E^-]$ containing~$P_1$, and therefore 
	$I(P_1)\cap I(C \setminus P_1)=\emptyset$, in particular  
	$I(P_1)\cap I(P_2)=\emptyset$.
	
	If $P_1\cup P_2=C$, then we are done, otherwise let  $P_3:= C\setminus (P_1\cup P_2)$. Clearly, $I(P_3)$ is   disjoint from $I(P_1)$ and $I(P_2)$, because $I(P_i)\cap I(C\setminus P_i)=\emptyset$ for $i=1,2$. Moreover, defining $I_1, I_2$ and~$I_3$ as in the statement of the lemma, $P_3$ is also a shortest path in $G_{I_3}$, as otherwise the partition of~$C$ into two paths $\{P_1\cup P_2, P_3\}$ would contradict (a), since by the maximal choice of $P_2$ we know that $P_1\cup P_2$ is also not a shortest path  in $G_{I_1 \cup I_2}$.
	Similarly, $P_1$ is also a shortest path in~$G_{I_1}$, as otherwise the partition of~$C$ into two paths $\{P_1, P_2\cup P_3\}$ would contradict (a).
\end{proof}

We are now ready to present the main result of this section.

\medskip
\noindent{\bf FPT-algorithm for SOC with parameter~$c$:}

\begin{itemize}[leftmargin=38pt,itemindent=16pt,parsep=4pt] 
	\item[{\bf Step 0.}] Normalize $w$, and initialize $\mathcal{Q}=\emptyset$.
	\item[{\bf Step 1.}] 
	For all $I\subseteq [c]$, compute a shortest $(x,y)$-path $P(x,y,I)$ in $G_I$ for all $x\ne y\in V$. 
	\item[{\bf Step 2.}] For all $u, v\in V$ with $u \ne v$:
	\begin{itemize}[parsep=2pt,topsep=1pt,]
		\item[{\bf (a)}] if $uv\in E^+$ then let  $R\leftarrow\{uv\}$, and perform (c).
		\item[{\bf (b)}] if $u, v\in V(K_j)$ for some $j\in[c]$ then let $R\leftarrow K_j(u,v)$, and perform (c).
		\item[{\bf (c)}] For all $x\in V$, 
		$I_u\subseteq [c]\setminus I(R)$ and $I_v= [c]\setminus (I_u \cup I(R))$: \\[2pt]
		if $Q=R\cup P(u,x,I_u)\cup P(x,v,I_v)$ is an odd cycle, then add $Q$ to~$\mathcal{Q}$.
	\end{itemize}
	\item[{\bf Step 3.}] If $\mathcal{Q}\neq \emptyset$, then  {\bf return} $Q \in \mathcal{Q}$ with the minimum weight; otherwise {\bf return} ``There is no odd cycle in $G$''.
\end{itemize}

\noindent
{\bf Running time.}
Step~0 can be performed in linear time, as explained in Section~\ref{sec:prelim}. 
Step~1 computes shortest paths for all pairs of vertices in $2^c$ different graphs with conservative weights. By Theorem~\ref{thm:classical}, the SP problem can be solved in polynomial time;
the book by Korte and Vygen describes an $O(n^4)$ time algorithm~\cite[Theorem 12.14]{KV} for this problem.
Step~2  has $O(n^2)$ iterations, and inside each of these (c) in turn checks for at most $2^c n$ edge sets whether it forms an odd cycle. As this takes $O(n)$ time for each set, Step~2 takes altogether $2^c O(n^4)$ time, so this is the total time used by the
 FPT-algorithm:

\begin{theorem}\label{thm:alg-fpt} The above algorithm returns a $w$-minimum odd cycle, if $G$ is non-bipartite, and its running time is $2^cO(n^4)$. 
\end{theorem}

\begin{proof} We have already proved the assertion on the complexity, so let us prove the correctness of our algorithm.
	
	If $G$ is bipartite,  $R \cup P(u,x,I_u)\cup P(x,v,I_v)$ of Step~2 is even for all possible choices,  since it is a closed walk in a bipartite graph. So $\mathcal{Q}$ remains empty, and the algorithm returns a correct answer. Otherwise let $C$ be a shortest odd cycle; we show that the algorithm puts into~$\mathcal{Q}$ an odd cycle of the same weight as $C$, and thus returns a correct solution.
	
	By the final assertion of Lemma~\ref{lem:notshortest}, $C$ admits a compact partition $P_1,\dots,P_m$, where $P_1$ either consists of an edge $uv\in E^+$, corresponding to a choice in Step~2(a) of our algorithm, or $P_1=C\cap K_j$ for some $j\in [c]$, that is, $P_1$ is a negative path between two distinct vertices~$u$ and~$v$ in $K_j$, corresponding to a choice in Step~2(b). Hence, at least once in Step~2 the path~$P_1$ gets chosen as~$R$. Recall that in our compact partition, {\em any other path~$P_i$  ($2 \leq i \leq m\le 3$)  in the partition is shortest in $G_i$} where $G_i=G_{I(P_i)}$.
	
	Now if $m=2$, then $P_2$ is a shortest  $(u,v)$-path in $G_2$. Consider the choice of Step~2(c) for~$P(u,x,I_u)$ with $x:=v$ and $I_u:=I(P_2)$. Since with these choices $P(u,x,I_u)=P(u,v,I_2)$ is also a shortest $(u,v)$-path in $G_2$, we get $w(P(u,x,I_u))=w(P_2)$. Since $w$ is normal, we also know that $P(u,x,I_u)$ has the same parity as $P_2$.
	Since now $P(x,v,I_v)=P(v,v,I_v)$ is a trivial path independently of the choice of $I_v$, we get that $w(P(x,v,I_v))=0$. 
By claim~(b) of Lemma~\ref{lem:notshortest}, we also know that $P_1$ and $P(u,x,I_u)$ are  openly disjoint,
 thus $Q$ is a cycle. 
  Moreover, the weight of~$Q=P_1 \cup P(u,x,I_u) \cup P(v,x,I_v)$ is $w(P_1)+w(P_2)=w(C)$. Hence,
  $Q$  has the same  weight and then by normality also the same parity as $C$, and is therefore a $w$-minimum odd cycle contained in~$\mathcal{Q}$, as claimed.
	
	If $m=3$, then the shortest path $P_2$ in~$G_2$ and the shortest path $P_3$ in~$G_3$ have a common endpoint, denote it by~$x$. Again, setting $I_u:=I(P_2)$ we get that $I_v=[c] \setminus (I(P_1) \cup I(P_2))=I_3$ also holds by our definitions. Moreover, $P(u,x,I_u)$ and $P(x,v,I_v)$  have the same weight and, by the normality of~$w$, the same parity as $P_2$ and $P_3$, respectively. Applying  claim~(b) of Lemma~\ref{lem:notshortest} to 
 the paths~$P_1$, $P(u,x,I_u)$, and $P(x,v,I_v)$, we get that they are mutually openly disjoint. Hence, we can conclude again that $Q=P_1 \cup P(u,x,I_u) \cup P(v,x,I_v)$ is a cycle, and has the same weight and parity as $C$. 
\end{proof}

It is easy to see that the  $w\ge 0$ special case of the FPT-algorithm consists only of
$n$ shortest path computations and does not rely on matchings. (Indeed, then the enumeration of the components of $E^-$ disappears, and  one execution of Dijkstra's shortest path algorithm  computes a shortest path from a given vertex to any other.)  This is not surprising, since it is well-known that an odd walk can be determined by $n$ shortest path computations in an auxiliary graph, both for undirected and directed graphs (see Proposition~\ref{prop:non-neg}).  
The same method is not suitable for determining shortest even cycles, since the proof of Lemma~\ref{lem:notshortest} relies on symmetric differences and $\emptyset$-joins, and heavily uses the fact that a shortest odd $\emptyset$-join contains a shortest odd cycle, while a shortest even $\emptyset$-join is
the empty edge set, having weight~$0$. In undirected graphs shortest even cycles 
for non-negative weights can be of course determined by solving SOP$_+$ problems (solvable in polynomial time according to Theorem~\ref{thm:classical})
for the endpoints of edges. However, for directed graphs the problem is more difficult (see in Section~\ref{sub:classic}, under the 
paragraph ``Digraphs'').  

\smallskip

Theorem~\ref{thm:alg-fpt} has the   immediate consequence  for MOTJ$_+$ that   after computing a $w$-minimum $T$-join~$F$,  {\em a $w$-minimum odd $T$-join can be computed in  $2^{c}O(n^4)$ time where $c$ denotes the number of connected components of $F$. }

Indeed, computing $F$ takes  $O(n^3)$ time for any~$T$, see~\cite[Section 29.2]{SYB}. If   $F$ is odd,   we are done. If not, by Proposition~\ref{prop:reductocycle} the minimum odd $T$-join problem is equivalent to determining a $w[F]$-shortest odd cycle $C$ in the graph $G$, and the set of negative edges of $w[F]$ has at most~$c$ negative components.

\iffalse
\begin{proof}
	We start by computing a  $w$-minimum $T$-join~$F$, which fits into $O(n^3)$ time for any~$T$, see~\cite[Section 29.2]{SYB}. If $F$ is odd, then we are done. If not, by Proposition~\ref{prop:reductocycle} the minimum odd $T$-join problem is equivalent to determining a $w[F]$-shortest odd cycle $C$ in the graph $G$.  Note that any inclusionwise minimal $T$-join consists of at most~$|T|/2$ connected components, $F$, and therefore the set of negative edges of $w[F]$ has at most~$|T|/2$ negative components. Thus the assertion of our corollary immediately follows from Theorem~\ref{thm:alg-fpt}. 
\end{proof}
\fi 

Since any inclusionwise minimal $T$-join consists of at most~$|T|/2$ connected components, we also obtain the following:

\begin{corollary}\label{cor:alg-fpt} 
	Given an instance  $(G,w,T)$ of MOTJ where $w$ is non-negative, a  $w$-minimum odd $T$-join (if it exists) can be computed in $2^{|T|/2}O(n^4)$ time.
\end{corollary}

As already mentioned, the fact  that MOTJ$_+$ can be solved in FPT time parameterized by~$|T|$ has already been proved by Cook, Espinoza and Goycoolea~\cite[Proposition 5.3]{CEG}.
Their algorithm runs in time $O(2^{|T|}+|T|^2n^2+n^3)$, so its dependence on $n$ is better than in  Corollary~\ref{cor:alg-fpt}, but  their dependence on $|T|$ is slightly worse.

\section{\texorpdfstring{$\NP$}{NP}-completeness}
\label{sec:NPC}

We present now a well-known $\NP$-complete problem that will be reduced to SOCp. Its planar special case is known to be one of the simplest open disjoint paths problems. 

\medskip\noindent
{\bf BACK AND FORTH PATHS (BFP)}

\noindent {\bf Input}:  A digraph $\hat G=(\hat V, \hat E)$ and $s\ne t \in \hat V$.

\noindent {\bf Question:}  Are there two openly disjoint paths, one from $s$ to $t$, the other from $t$ to $s$?

\begin{theorem}[\!{\cite[Theorem 2]{BFP}}, see also {\cite[p.~1225, footnote~6]{SYB}}]
	\label{thm:BFP}
	BFP is $\NP$-complete. 
\end{theorem}

Before proving the main  $\NP$-completeness results we are interested in, it will be useful to deduce the $\NP$-completeness of the directed versions of SOCp$_+$,  SOP$_+$, that immediately follow  from this theorem, and  already for non-negative weights:

\begin{proposition}[\!\!\cite{LP,Thomassen85}]
	\label{prop:dir-hardness}
	The directed variants of the SOCp$_+$, SECp$_+$, SOP$_+$ and SEP$_+$ problems are all $\NP$-complete.
\end{proposition}

The proof of Lapaugh and Papadimitriou, that of Thomassen and ours were found independently:  the $\NP$-complete problems used for the reductions  slightly differ from one another, but they are all from \cite{BFP}. We include our version of the proof  to show, in a simpler situation, the  starting step of our $\NP$-completeness proof for undirected graphs.

\begin{proof}
	The directed variant of SOCp$_+$ is NP-complete, because 
	given an instance $(G,s,t)$ of BFP,
	subdividing each edge of $G$ and then splitting $t$ into an in-copy~$t_{\rm in}$ and an out-copy~$t_{\rm out}$ in the usual way (with all incoming edges arriving at $t_{\rm in}$ and all outgoing edges leaving from~$t_{\rm out}$, and with a new edge~$t_{\rm in}t_{\rm out}$), 
	there exists an odd cycle going through~$s$ in the constructed digraph if and only if there is a pair of back and forth paths between $s$ and $t$ in the original digraph.
	Now, the directed version of SEP$_+$ is also $\NP$-complete, since finding a shortest odd cycle through~$s$ can be done by finding a shortest even path from $s$ to an in-neighbor of $s$. Clearly, the directed variants of SOP$_+$ and SEP$_+$ are equivalent, because we can flip the parity of all paths starting at a vertex~$s$ by subdividing each edge leaving~$s$; the same trick shows that the directed variants of SOCp$_+$ and SECp$_+$ are equivalent. 
\end{proof}

The proof shows that already the existence versions of the problems in Proposition~\ref{prop:dir-hardness} are $\NP$-hard. 
However, we remark that for planar graphs the complexity of BFP is open \cite[p.~1225, footnote~8]{SYB} and so seems to be the complexity of~SOCp for conservative planar undirected graphs
or SOC$_+$ for directed planar graphs.

The polynomial-time solvability of SOC for non-negative weights is straightforward (Proposition~\ref{prop:non-neg}), but SOC in conservative directed graphs is a more difficult problem because neither the tentative generalizations of Lemmas~\ref{lem:onepathgen} and \ref{lem:notshortest}  hold for directed graphs. There is also no relevant indication  that these problems could be polynomial-time solvable,  contrary to the undirected case. Are they $\NP$-hard?  

\smallskip
We now focus on undirected graphs, and switch to the statement and proof of one of our main messages:

\begin{theorem}\label{thm:SOP}  
	SOCp$_{\pm 1}$ is $\NP$-complete, even when the negative edges form a matching, $k=1$,   and there exists  a vertex~$t$ so that $G-t$ is bipartite. 
\end{theorem}

\begin{proof}
	SOCp$_{\pm 1}$ is clearly in $\NP$. Let the digraph $\hat G=(\hat V,\hat E)$ with vertices $s, t\in \hat V$ be an instance of BFP, and construct from  it an undirected graph as follows. Split each vertex $v\in \hat V\setminus\{t\}$ to an {\em out-copy} $v_1$ and an {\em in-copy} $v_2$, except for leaving $t$ as it is,  but defining $t_1 \leteq  t_2 \leteq  t$. For  each arc $uv\in \hat E$ define an  edge $u_1v_2$ with $w(u_1v_2) \leteq  1$. Furthermore, add an edge $v_1v_2$ for each  $v\in \hat V\setminus \{t\}$ with $w(v_1v_2) \leteq  -1$.
	
	Denote $V_i \leteq  \{v_i: v\in \hat V\}$ for $i=1,2$, and $E \leteq  \{u_1v_2: uv\in \hat E\}\cup\{v_1v_2: v\in \hat V\setminus\{t\}\}$, so that  the constructed (undirected) graph is $G=(V_1\cup V_2, E)$, and let $k \leteq  1$. 
	Clearly, the negative edges form a matching, and thus the weight function~$w$ is conservative. Note that $G-t$ is {\em bipartite, so all odd cycles contain~$t$.}  
	
	\smallskip\noindent{\bf Claim}:    $\forall\,\hat C\subseteq\hat E$ cycle, $s,t\in V(\hat C)$,  $\exists\, C$ cycle in $G$, $w(C)=1$, $s_1\in V(C)$,  and vice versa. 
	
	\smallskip Indeed, let $\hat C\subseteq\hat E$ be a (directed) cycle in~$\hat G$ with $s,t\in V(\hat C)$,  and let us associate with it the (undirected) cycle $C \leteq  \{u_1v_2: uv\in \hat C\}\cup\{v_1v_2: v\in V(C) \setminus \{t\}\}$ in $G$. The cycle $C$ alternates between edges of weight $1$ and $-1$ in every vertex but $t$, so $w(C)=1$,  and $s_1\in V(C)$.    
	
	Conversely, a cycle $C\subseteq E$ in $G$ with $w(C)=1$ and $s_1\in V(C)$ must be an odd cycle due to its weight, so $t\in V(C)$ follows as noted earlier. Moreover, $C-t$ must alternate between edges of weight $-1$ and $1$, so $C$ corresponds to a directed cycle $\hat C\subseteq\hat E$. Since $s_1\in V(C)$ by definition, and we know $t\in V(C)$ as well,  the  cycle  $\hat C$ contains both $s$ and $t$, so the claim is proved.
	
	\smallskip The claim shows that our construction reduces BFP to SOCp$_{\pm 1}$, since a solution of BFP is exactly a cycle $\hat C\subseteq\hat E$ in~$\hat G$ with $s,t\in V(\hat C)$,  and according to the claim such a cycle exists if and only if there exists an odd cycle~$C$ in~$G$ of weight at most~$1$ containing $s_1$; note that an odd cycle of weight at most~1 can have neither weight~0 (due to its parity) nor negative weight (due to  conservativeness), so must have weight exactly~1. The instance $(G, p\leteq s_1, k\leteq 1)$ of SOCp$_{\pm 1}$ to which BFP is reduced satisfies the additional assertions, as checked above, so we can conclude that SOCp$_{\pm 1}$ is $\NP$-complete and already for the family of the claimed particular instances. 
\end{proof}

By simply inspecting the instances of the above proof, the $\NP$-hardness of the following problem of  Lov\'asz (\!\!\cite[Open Problem~27, pp. 517]{SYB}) is an immediate corollary.

\begin{corollary}\label{cor:SOP}  SOP$_{\pm 1}$ is $\NP$-complete, even when the negative edges form a matching,  $k=1$,  and there exists a vertex~$t$ so that $G-t$ is bipartite.
\end{corollary}
\begin{proof} 
	SOCp is the special case of SOP where $s=t$, so we are done. 
	If we want to require $s\ne t$, then with the notation of the proof of Theorem~\ref{thm:SOP}, observe that the instance $(G, s_1, k=1)$  of SOCp  has a `yes' answer if and only if there exists an odd $(s_1,s')$-path of weight $k=1$ in the graph~$G'$ obtained from~$G$ by replacing the edge $s_1s_2$ with an $s's_2$ edge of weight $-1$ for a  new vertex~$s'$. 
\end{proof}

Note that the reduction keeps planarity, but  the complexity of BFP is open for planar graphs, so we do not know the complexity of  SOCp$_+$  for planar graphs. 

\smallskip
Let us now consider the following problem which has a strong, although not immediately straightforward, relationship with the problems we study. 

\medskip
\noindent {\bf DISJOINT SHORTEST PATHS IN CONSERVATIVE GRAPHS (DISP)}

\noindent {\bf Input}: An undirected, conservative graph $G=(V,E)$ with $w:E\rightarrow\{1,-1\}$, $s_1,s_2, t_1, t_2\in V$, and $k\in\mathbb{Z_+}$.

\noindent {\bf Question:} Does $G$ contain two vertex-disjoint  $(\{s_1,s_2\},\{t_1, t_2\})$-paths with total weight at most~$k$? 

\smallskip 
While DISP for non-negative weights is a special case of the well-known minimum cost flow problem, and it is so for conservative digraphs as well, {\em it seems the question has not even been asked for conservative undirected graphs!}  For these, a tentative reduction to digraphs meets the same  obstacle we met for shortest paths in the Introduction (Section~\ref{sec:intro}): directing an edge with negative weight in both directions creates a negative cycle consisting of two arcs.  However, although the undirected shortest path problem (SP) is still solvable in polynomial time even if the methods are more difficult than those applied for directed graphs, this is not the case for DISP. It  turns out to be $\NP$-complete, and essentially for the same reason as SOP or SOCp:  

\begin{corollary}\label{cor:DISP} DISP is $\NP$-complete, even when the negative edges form a matching, and $G$ is bipartite.
\end{corollary} 

\begin{proof}
	We reduce from BFP using the same construction as in the proof of Theorem~\ref{thm:SOP}
	with the only difference that we split all vertices of the input digraph~$\hat G=(\hat V,\hat E)$, including~$t$, add the edge $t_1t_2$ to $E$, and define  $w(t_1 t_2)\leteq -1$.
Then the resulting graph~$G$ is bipartite, and
$(\hat G,s,t)$ is a `yes'-instance of BFP 
if and only if  there exists a cycle~$C$ of weight~$0$ in~$G$ containing both $s_1$ and $t_1$, which in turn holds 
if and only if there exist two vertex-disjoint $(\{s_1,s_2\},\{t_1, t_2\})$-paths of total cardinality~$k=2$ in~$G$.
\end{proof}
Corollary~\ref{cor:DISP} contrasts the well-known fact that 
finding two disjoint $(\{s_1,s_2\},\{t_1,t_2\})$-paths for some vertices $s_1$, $s_2$, $t_1$, and $t_2$ with minimum total weight
in an undirected graph with non-negative edge weights is a standard classical minimum-cost flow problem \cite{SYB}. 
The example depicted in Figure~\ref{fig:interlaced} gives some intuition on the strong connection between  DISP and our problems SOCp and SOP.

\begin{figure}[t]
\centering
\includegraphics[scale=1]{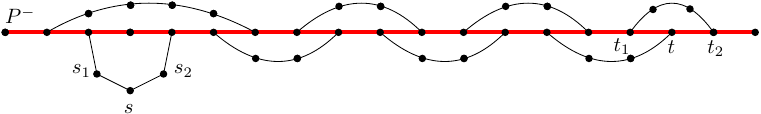}
\caption{An example where edges of the path~$P^-$ have weight~$-1$, shown as red, bold lines, with all remaining edges having weight~1. An odd cycle containing $s$ must also contain $t$, and the unique such cycle yields also a solution for DISP (with vertices $s_1,s_2,t_1,t_2$), and also a shortest odd $(s_1,s_2)$-path. A shortest odd  $\{s_1,s_2\}$-join consists of an $(s_1,s_2)$-path of $4$ edges and weight~$0$, and the odd cycle of~$5$ edges through $t_1$, $t$, $t_2$, altogether  $9$ edges with total weight $1$.
}
\label{fig:interlaced}
\end{figure}

\section{Connections, Questions and Conclusion}\label{sec:conn}

In this section we establish further connections between  the problems we have been studying to some known results and open questions.

\subsection{Classical Results}
\label{sub:classic}

\smallskip\noindent{\bf Forgetting the parity:} We mentioned  in Section~\ref{sec:intro} the simple fact that {\em  a shortest $(s,t)$-path in a conservative graph  (SP) can be determined by finding an inclusionwise minimal  shortest $\{s, t\}$-join, which is in fact equivalent to the minimum-weight $T$-join problem for arbitrary, not necessarily conservative weights.} 
The first, well-known solution of the latter problem by Edmonds~\cite{E} reduces it to  non-negative weights, and then solves it as a weighted matching problem on $T$, see also~\cite[Section~29.2]{SYB}. 

\smallskip\noindent{\bf Non-negative weights:}  For polynomial algorithms to various problems see Theorem~\refeq{thm:classical}. Note that SOC$_+$ is much easier than SOC 
 for both undirected and directed graphs: there is a well-known method for solving it via $n$ shortest path computations in an auxiliary graph. 
We state and prove the exact complexity result for comparison and further reference in Proposition~\ref{prop:non-neg} below.
Strangely, the Odd Path Polyhedron (the ``dominant'' of odd paths, and the related integer minimax theorem~\cite{S2}, see also \cite[Chapter 29.11e]{SYB}) have been determined much later.

\begin{proposition}[\!\!{\cite[Chapter 8.3]{GrLov88}}]
\label{prop:non-neg}
Give an undirected or directed graph with non-negative weights, a shortest odd cycle  can be determined with at most $n$ executions of Dijkstra's algorithm, that is, in  $O(mn+n^2\log n)$ time. 
\end{proposition}

\begin{proof}
Let $G=(V,E)$ be the input graph with edge-weight function $w$. If $G$ is directed, then double the vertex set and the edge set of~$G$ by taking two distinct copies $v_1$ and $v_2$ of each $v\in V$, and for each $uv\in E$ adding  edge $u_1v_2$ with the same weight as $uv$.  It is well-known and easy to see that the  shortest among all $(v_1,v_2)$-paths, $v\in V$, in the resulting (undirected) graph yields a shortest odd cycle in~$G$. 
The complexity of finding $n$ shortest paths for non-negative edge-weights takes $n$ executions of Dijkstra's algorithm. 

For undirected graphs the problem can be reduced to directed graphs by taking each edge in both directions.
\end{proof}
Even though SEC$_+$ is slightly more difficult than SOC$_+$ in undirected graphs it can obviously be solved with $|E(G)|$ shortest odd path computations SOP$_+$, solved in polynomial time (Theorem~\refeq{thm:classical}).

\smallskip\noindent{\bf Digraphs:}  In Section~\ref{sec:NPC} (mainly after Theorem~\ref{thm:BFP}) we mentioned the complexity for digraphs of the problems analogous to those we are studying. Similar reductions work for the shortest odd and even cycle problems through a given vertex (or equivalently, an edge) in a directed graph, proving that these problems are $\NP$-hard~\cite{LP},  \cite{Thomassen85}, see Proposition~\ref{prop:dir-hardness}.
While for undirected graphs the shortest odd or even cycles are both similarly easy  to determine if all weights are non-negative, this is not the case for directed graphs. 
As we have seen in Proposition~\ref{prop:non-neg}, SOC$_+$ for directed graphs is as easy as in undirected graphs. 
By contrast,  SEC$_+$ for directed graph is inherently more difficult, and its complexity is not completely settled: finding any even  (directed) cycle has been an open problem for more than two decades, before solved by Robertson et.\ al.~\cite{RST99} and McCuaig~\cite{McCuaig} independently, 
and the problem of finding a shortest even cycle has been solved very recently by Bj\"orklund et al.~\cite{BHK}, but  only for unweighted digraphs and with a randomized algorithm.  

Two problems more closely related to our work also remain open: 

\begin{problem}
What is the complexity of SOC in digraphs with conservative weights?
\end{problem}

The feelings are not really oriented towards polynomial-time solvability, since nothing similar to Lemma~\ref{lem:onepathgen} seems to be true, driving the search towards enumeration.  

\begin{problem}
Is  SOCp$_+$ polynomially solvable in planar directed graphs? More generally, what is the complexity for planar graphs of the problems proved to be $\NP$-hard (for undirected and directed graphs) in this article?
\end{problem}
The source of this questioning is that BFP is open for planar graphs \cite[p.~1225, footnote~8]{SYB}.  

\medskip\noindent {\bf Odd $T$-joins}: Their properties with respect to packing and covering have been intensively studied in terms of the ``idealness'' (integrality) of their blocking polyhedra. Idealness roughly means that  good characterization (minimax) theorems hold for the minimization of odd $T$-joins {\em for non-negative weight functions}. 

The corresponding algorithms and complexity results have been analyzed in Section~\ref{sec:Tractable}, where we anticipated that non-negativity is not an essential condition in this case. We provide a precise proof  below for the equivalence of arbitrary weight functions with non-negative ones (Theorem~\ref{thm:equiv}).

\smallskip\noindent {\bf Max Cut}:
The ``min side" of the mentioned minimax theorems concerns transversals of odd $T$-joins
which, in the simplest case of odd cycles (i.e., $T=\emptyset$),
are easily seen to be exactly the complements of cuts:  their minimization is equivalent to the \textsc{Maximum Cut} problem, one of the sample $\NP$-hard problems. However, for planar graphs the duality between faces and vertices  reduces this problem to the shortest $T$-join problem \cite{Bar80}, solving \textsc{Maximum Cut} for planar graphs;    for graphs  embeddable into the projective plane  the corresponding reduction is to MOTJ, and only a partial solution could be given to the corresponding special case of MOTJ \cite{MaxCut}. 

\smallskip\noindent {\bf SOC versus SOCp}: 
SOC can clearly be reduced to SOCp but the opposite reduction seems to organically resist. This is an analogous situation to the problem of finding a minimum-weight odd {\em hole} (an induced cycle of cardinality at least four) through a given vertex is $\NP$-complete~\cite{Bienstock}, while without the requirement of containing a given vertex it has been recently proved to be polynomially solvable~\cite{CS3}. 

\medskip
The applications and relevance of the SOC and MOTJ problems and signs of their tractability, mentioned in Section~\ref{sec:Tractable} and leading to Conjecture~\ref{conj:SOC}, 
makes it interesting to clarify their polynomial equivalence, which we do in Section~\ref{sub:oddT-join}.

\subsection{Equivalence of SOC and MOTJ}\label{sub:oddT-join}

In this section we show that weighted optimization problems on odd $T$-joins  
are actually polynomially equivalent to their special case for conservative weight functions, which in turn can be shown to be equivalent to the case where $w$ is restricted to be non-negative, or $T$ to be empty.

\begin{theorem}\label{thm:equiv}
The following problems are polynomially equivalent:
\begin{itemize}
	\setlength\itemsep{-2pt}
	\item[(i)] \vspace{-4pt} MOTJ;
	\item[(ii)] MOTJ with conservative weights;
	\item[(iii)] MOTJ$_+$;
	\item[(iv)] MOTJ with conservative weights for  $T=\emptyset$; 
	\item[(v)] SOC with conservative weights. 
\end{itemize}
\end{theorem}

\begin{proof}  A polynomial algorithm for $(i)$, i.e., MOTJ in general, clearly implies one for $(ii)$, which,  in turn, implies one for $(iii)$.

To prove the polynomial-time solvability of $(iv)$ from that of $(iii)$, consider the input of MOTJ with $T=\emptyset$ consisting of a graph~$G=(V,E)$ and a conservative $w$. %that can take arbitrary values, solely,  under the constraint of conservativeness. W.l.o.g. 
We can assume that $G$ contains an even number of negative edges, since otherwise we can simply add to~$G$ an edge of weight~$-1$ incident to a new vertex. 
Define now a non-negative weighted  instance~$(G,|w|,T)$ of MOTJ with  $T\leteq \{v\in V: d_{E^-} (v) \hbox{ is odd} \}$ where $E^-\leteq \{e\in E: w(e)<0\}$. Then $E^-$ is a $|w|$-minimal $T$-join, and it is even. Now by Proposition~\ref{prop:reductocycle}, $J$ is a $|w|$-minimal odd $T$-join if and only if $C\leteq J\Delta E^-$ is a $w=|w|[E^-]$-minimal odd $\emptyset$-join. Hence, an algorithm for~$(iii)$ applied to~$(G,|w|,T)$ yields a solution for our instance~$(G,w)$ of~$(iv)$.

The claim that polynomial-time solvability of~$(iv)$ implies the same for~$(v)$ follows by noting that a solution for~$(v)$ can be obtained from a solution for~$(iv)$ with the same input instance by deleting the $0$-weight even cycles, and possibly all but one $0$-weight odd cycle.

We have thus asserted the path of implications from the polynomial-time solvability of $(i)$ to that of $(v)$.
A polynomial-time algorithm for~$(i)$ follows from one for~$(v)$ by Proposition~\ref{prop:reductocycle}, since a shortest odd cycle for conservative weights is always a minimum-weight odd $\emptyset$-join.
\end{proof}

Note that (iv) is a specialization of (ii) to the  special case $T=\emptyset$, and actually any special $T$ can play this role of implying a polynomial solution to (i)-(v). (Indeed, to prove  this,  one only has to change   $w$ to $w[J]$ in the proof of (iv),  where $J$ is a $w$-minimum $T$-join, and then $E^-$ is changed to $E^-\Delta J$. In the proof we had $T=\emptyset$ and $\emptyset$ is a shortest $\emptyset$-join in a conservervative graph.)  Figure~\ref{prop:dir-hardness} illustrates how much easier it is to find a shortest $\{s,t\}$-join than to find a shortest $(s,t)$-path in a graph with a conservative weight-function.

Restricting  MOTJ with  $|T|\le 2$ and requiring at the same time non-negative weights results in an easy problem, as we have shown in Section~\ref{sub:algoT2}.
However, if only one of $T$ and~$w$ is restricted, then the general problem can be reduced to these (seemingly) more special ones, as stated by Theorem~\ref{thm:equiv}. The case where the absolute values of the weights are  $1$ are not proved to be essentially easier than   general weights,  for any of the problems.

\subsection{Conclusion}
\label{sub:conc}
The MOTJ problem is a relevant combinatorial optimization problem that may be solvable in polynomial time. The complexity of this problem remains open, but we proved that SOC, polynomially equivalent to MOTJ$_+$ or  MOTJ, 
is fixed-parameter tractable when parameterized by the number of components of negative edges. 
If negative weights are also allowed, then finding a minimum-weight odd  $\{s,t\}$-join is already equivalent to general MOTJ. 

We also proved that the related SOCp, SECp, SOP, SEP, and DISP problems in conservative undirected graphs are  $\NP$-complete,  answering a long-standing question of  Lov\'asz \cite[Problem~27]{SYB}, and we exhibited some related, polynomial algorithms.  At the same time we pointed at three  open  challenges for undirected and directed graphs, one of which is MOTJ itself. 

Another interesting research direction is now to study the parameterized complexity and approximability of the SOP problem both for directed and undirected graphs and its other $\NP$-hard variants. Some initial FPT results  have been achieved by part of our research group formed during the 12th Eml\'ekt\'abla Workshop, G\'ardony, Hungary, 2022~\cite{GardonyGroup}. 

\subsection*{Acknowledgement}
The collaboration of the authors has been hosted and supported by the 12th Eml\'ekt\'abla Workshop, G\'ardony, Hungary, July 2022; we would like to express our gratitude for the organizers.  
Also, we would like to thank Yutaro Yamaguchi for initiating a research group on the problem of finding shortest odd paths in conservative graphs, and Alp\'ar J\"uttner, Csaba Kir\'aly, Lydia Mendoza, Gyula Pap as well as Yutaro for the enjoyable discussions. 
Ildik\'o Schlotter acknowledges the support of the Hungarian Academy of Sciences under its Momentum Programme (LP2021-2), and the Hungarian Scientific Research Fund (OTKA grants~K128611 and K124171).

\end{document}